%% file: paper_arxiv_v1_nc.tex
\documentclass[abstract=true, DIV=14]{scrartcl}

\usepackage[noEnd,commentColor=black]{algpseudocodex}
\usepackage{amsmath, mleftright}
\usepackage{amssymb}
\usepackage{amsthm}
\usepackage[mathscr]{euscript}
\usepackage[shortlabels]{enumitem}
\usepackage{graphicx} 
\usepackage{subcaption}
\usepackage{thmtools}
\usepackage{todonotes}
\usepackage{xcolor}
\usepackage{microtype}
\usepackage[T1]{fontenc}
\usepackage{nicefrac,xfrac}

\makeatletter\input{t1cmss.fd}\makeatother
\DeclareFontShape{T1}{cmss}{bx}{n}
  {<5><6><7><8><9><10->ecsx1000}{}
  
\makeatletter
\def\@fnsymbol#1{\ensuremath{\ifcase#1\or \!\;\or \!\;\or \ddagger\or
   \mathsection\or \mathparagraph\or \|\or **\or \dagger\dagger
   \or \ddagger\ddagger \else\@ctrerr\fi}}

\makeatother

\usepackage[normalem]{ulem}

\usepackage[english]{babel}

\usepackage{tcolorbox}

\usepackage{upgreek}
\usepackage[colorlinks]{hyperref}
\usepackage[capitalize]{cleveref}

\crefname{equation}{eq.}{eqs.} 
\crefname{enumi}{}{} 

\crefname{icase}{case}{cases}
\crefname{ipart}{part}{parts}
\crefname{iprop}{property}{properties}
\crefname{iinv}{invariant}{invariants}

\crefformat{section}{\S\,#2#1#3}
\crefformat{subsection}{\S\,#2#1#3}
\crefrangeformat{section}{\S\S\,#3#1#4--#5#2#6}
\crefmultiformat{section}{\S\S\,#2#1#3}{,~#2#1#3}{,~#2#1#3}{,~#2#1#3}

\usepackage{authblk}
\usepackage{float}

\usepackage[normalem]{ulem}

\usepackage{tikz}
\usetikzlibrary{calc}

\newcommand{\R}{\mathbb{R}}
\newcommand{\fO}{\mathcal{O}}

\newcommand{\F}{\mathcal{F}}
\newcommand{\cF}{\mathcal{F}}

\newcommand{\Unif}{\operatorname{Unif}}
\newcommand{\HH}{\mathsf H}
\newcommand{\Hc}{\mathsf H_{\mathrm{cont}}}
\newcommand{\II}{\mathsf I}
\newcommand{\1}{\mathbf 1}
\newcommand{\Dent}{D^{\mathrm{ent}}}
\newcommand{\hbin}{\mathsf H_2}

\newcommand{\Ebb}{\mathbb E}
\newcommand{\Pbb}{\mathbb P}
\newcommand{\Qbb}{\mathbb Q}

\setkomafont{captionlabel}{\bfseries}
\renewcaptionname{english}{\figurename}{Fig.}
\setcapindent{0pt}
\renewcommand{\alpha}{\upalpha}

\newcommand{\cT}{\mathcal{T}}

\newcommand{\cS}{\mathcal{S}}

\newcommand\utimes{\mathbin{\ooalign{$\cup$\cr
   \hfil\raise0.42ex\hbox{$\scriptscriptstyle\times$}\hfil\cr}}}
\newcommand\bigutimes{\mathop{\ooalign{$\bigcup$\cr
   \hfil\raise0.36ex\hbox{$\scriptscriptstyle\boldsymbol{\times}$}\hfil\cr}}}

\newcommand{\connected}[1]{\def\temp{#1}\ifx\temp\empty\sim\else\overset{#1}{\sim}\fi}

\tikzset{
	point/.style={circle, fill, inner sep=1.5pt},
	smallpoint/.style={point, inner sep=1.2pt},
	tinypoint/.style={point, inner sep=1pt},
	hlbox/.style={fill, {white!90!black}},
	subrect/.style={draw, fill={white!80!cyan}}, 
	msrect/.style=subrect 
}

\newtheorem{theorem}{Theorem}[section]

\newtheorem*{theorem*}{Theorem} 
\newtheorem{lemma}[theorem]{Lemma}
\newtheorem{proposition}[theorem]{Proposition}

\theoremstyle{definition}

\title{Optimal chain density, entropy, and space-time tradeoffs for the TSP}

\author[1]{Alexandr Andoni\thanks{$^1$Email: \href{mailto:andoni@cs.columbia.edu}{andoni@cs.columbia.edu.}, \href{mailto:hantao.yu@columbia.edu}{hantao.yu@columbia.edu}}}
\author[2]{Justin Dallant\thanks{$^2$Email: \href{mailto:justin.dallant@tu-dresden.de}{justin.dallant@tu-dresden.de}, \href{mailto:laszlo.kozma@tu-dresden.de}{laszlo.kozma@tu-dresden.de}}}
\author[2]{L\'{a}szl\'{o} Kozma\protect\footnotemark[1]}
\affil[1]{Columbia University}
\affil[2]{Dresden University of Technology}
\author[1]{Hantao Yu\protect\footnotemark[1]}

\date{}

\begin{document}

\maketitle

\begin{abstract}

We nearly settle a natural extremal question  
about set systems over $[n]$: the tradeoff between the \emph{size} (number of sets) and the number of \emph{full chains}. 
This question was initially raised by Johnson, Leader, and Russell [Combin.~Probab.~Comp., 2015] as a counterpart to Sperner-type results in combinatorics.

Recently, a framework introduced by Ameli, Nederlof, and Wang, and independently by Dallant and Kozma [FOCS 2026] linked this question to the space- and time-complexity of Bellman-Held-Karp-style dynamic programming algorithms for permutation problems such as the traveling salesman (TSP). Precisely, they showed that a space-time product $\gamma^{n+o(n)}$ is feasible for the TSP, whenever a set system of (normalized) size $S$ and chain density $D$ exists, with $ \gamma = S^2/D$. In this paper we show an essentially \emph{optimal} bound of $\gamma \approx 3.1819$ for this quantity, closing the gap between the previous best lower and upper bounds of $\gamma \geq 3.015$ and $ \gamma \leq 3.572$ respectively. This implies a TSP algorithm with space-time product $\fO(3.1819^n)$ for input size $n$, as well as a limit to further improvements in this broad framework. More generally, we can obtain close to optimal values $D$ for any feasible value $S$, effectively settling the question of the number of full chains at every size. 

The crucial step towards our results is casting the extremal combinatorics question as an  \emph{information~vs.~entropy} tradeoff involving two random variables. This reformulation \emph{exactly} captures the optimal tradeoff for the combinatorial problem, leading to a framework in which primal-dual certificates can be derived, proving rigorous upper and lower bounds on $\gamma$. We also give a further application of our techniques, improving a bound of Duffus, Sands, and Winkler on the minimum size of fibres in the Boolean lattice.

\end{abstract}

\section{Introduction}
\label{sec1}

The traveling salesman problem (TSP) is one of the cornerstones of combinatorial optimization and has served as a proving ground for many of the algorithmic techniques invented in the past decades, whether exact, approximate, or heuristic, e.g., see~\cite{schrijver2005history, ApplegateEtAl2006, Cook2011, traub2024approximation}.  

Yet, despite significant effort, the fastest algorithm for solving TSP exactly, in the worst case, remains the Bellman-Held-Karp~\cite{HeldKarp1962,Bellman1962} dynamic program with running time and space $\fO^*(2^n)$. When the space budget is bounded by a polynomial, the fastest known algorithm is the Gurevich-Shelah~\cite{GurevichShelah1987} divide-and-conquer with running time $\fO^*(4^n)$.\footnote{TSP asks for a permutation of $n$ cities, minimizing the sum of pairwise distances between consecutive cities in the permutation. The mentioned results hold for arbitrary distance functions, but no better bounds are known even under a metric assumption. The $\fO^*(\cdot)$ notation is used to suppress polynomial factors.}  
A folklore combination of the two methods results in a running time $\fO^*(\cT^n)$ and space $\fO^*(\cS^n)$ with $\cT \cdot \cS = 4$, for various values of $1 \leq \cS \leq 2$ (e.g., see~\cite[\S\,10]{FominK10}). 

While the two extremes ($\cT = 2$ and $\cT=4$) of this tradeoff have resisted improvement for several decades now, in a surprising breakthrough in 2010, Koivisto and Parviainen~\cite{KoivistoParviainen2010} improved the tradeoff at intermediate points, in particular, to $\cT \cdot \cS \approx 3.93$, for $\cT \approx 2.7$. Their result builds on the classical dynamic programming (DP) approach of Bellman, Held, and Karp, but partitions the space of possible solutions (all permutations of the input) into sets of permutations that are linear extensions of certain partial orders. This allows searching each subset separately, reusing space between successive searches.   

Informally, the \emph{number of linear extensions} determines what fraction of the solution space is covered by one search, and thus the number of repetitions that are needed. The \emph{number of downsets} (poset ideals) exactly corresponds to the size of the DP table, dominating both the space- and time requirements of a single search. 
The efficiency of the scheme thus hinges on the existence of partial orders with few downsets and many linear extensions (relatively speaking).

These two goals are, of course, in tension; to support many linear extensions, a partial order must admit all of their prefix sets as downsets; the question is, intuitively, how much ``sharing'' is possible. Unfortunately, a tight quantitative relation between these two basic parameters of a poset is not known. Koivisto and Parviainen showed that a certain two-level poset class is optimal among a broader class of layered posets, and within this class, the most favorable tradeoff is attained at a concrete constant size of $n=26$. This construction yielded the above TSP-tradeoff of $\cT \cdot \cS \approx 3.93$.

\smallskip

Recently, Ameli, Nederlof, and Wang~\cite{ANW} and independently, Dallant and Kozma~\cite{DK} described a more general framework for space-time tradeoffs for TSP (as well as for a broader family of permutation problems), where the solution space is partitioned into sets of permutations that arise as \emph{maximal chains} of set families. The efficiency of the algorithm is then dependent on the number of maximal chains, and on the size of the family (the number of sets it contains).

This scheme generalizes the one by Koivisto and Parviainen; downsets of a poset $P$ form a set family whose maximal chains are exactly the linear extensions of $P$. It is, in fact, significantly more general, as there are $2^{2^n}$ set families, but only $2^{\fO(n^2)}$ partial orders over a ground set of size $n$~\cite{kleitman1975asymptotic}. The framework reduces the algorithm design question to one of extremal combinatorics:
\vspace{-0.1in}
{\begin{center}{\textbf{finding a set family of small size with many maximal chains}.} \end{center}}
\vspace{-0.05in}

More precisely,~\cite{ANW, DK} obtain the space-time tradeoff $\cT \cdot \cS = \gamma = S^2/D$, whenever there exists a set family over $[n]$ of size $S^n$ and number of maximal chains $D^n \cdot n!$. The quantities $S$ and $D$ can be seen as small constants, and it is sufficient to show the existence of a single finite set system with these parameters. (We give more precise definitions in \S\,\ref{sec11}.)   

The best construction given in \cite{DK} yields a value $\gamma \leq 3.572$ for a TSP space-time tradeoff $\cT \cdot \cS$ with the same value. In~\cite{ANW} the lower bound $\gamma \geq 3.015$ is shown via a delicate argument involving isoperimetric inequalities; a lower bound of $\gamma \geq 3$ more easily follows from a counting argument~\cite{DK, ANW}. 

The question of relating these two parameters of set families has also been raised before in extremal combinatorics. Johnson, Leader, and Russell~\cite{JLRsetSystems} asked which family, of a given size, maximizes the number of maximal chains (and hence their density). They conjectured that extremal families arise from two-level posets, closely related to the constructions found earlier by Koivisto and Parviainen~\cite{KoivistoParviainen2010}. 
The results of \cite{ANW, DK} refute this conjecture, yielding families with a more favorable tradeoff. The optimum value of $\gamma$ and the exact quantitative nature of the tradeoff between size and chain-density, however, was left open by these works. 

\paragraph{Our contribution.} In this paper we essentially close the gap, obtaining $3.1818 < \gamma < 3.1819$, where further refinement appears limited only by computation.  
To obtain these results, we reformulate the question as the tradeoff between two information-theoretic quantities involving a pair of random variables. We show that the optimal tradeoffs of the two settings \emph{coincide}: both upper- or lower bound certificates for one problem imply the same for the other. The advantage of the information-theoretic view, besides the simpler formulation, is that it is amenable to the efficient construction of such certificates. 
Upper bound certificates in the information-theoretic setting are, in turn, transferable back to the original setting, yielding set systems with the required parameters.

As a main algorithmic consequence of our results, 
we obtain an algorithm for the TSP with running time $\fO^*(\cT^n)$ and space $\fO^*(\cS^n)$ where $\cT \cdot \cS \leq 3.1819$, and we show that a value $\cT \cdot \cS \leq 3.1818$ is not attainable in this framework. 

Besides this point-wise result, we can, in principle, compute any point on the tradeoff curve between $S$ and $D$ with arbitrary accuracy. For easier interpretability, we also derive closed-form, approximate upper and lower bound curves on the optimal tradeoff.
This results, in the framework of~\cite{ANW,DK}, in TSP algorithms with an improved running time for any space budget $\fO^*(\cS^n)$, with $1 < \cS < 2$.

The number of maximal chains is a natural invariant of a set system, and hence its characterization likely has further applications. As an example, we show that our results imply a lower bound on the size of the smallest \emph{fibre} in the Boolean lattice (a fibre is a subfamily that intersects every maximal antichain). This result improves decades-old bounds by Duffus, Sands, and Winkler~\cite{DuffusSW90}, as well as by  {\L}uczak~\cite{DuffusSands2001}; see~\S\,\ref{sec14} and \S\,\ref{sec8}.

\subsection{The combinatorial problem}\label{sec11}

Recall that the Boolean lattice $2^{[n]}$ consists of all subsets of
$[n]=\{1,2,\ldots,n\}$, ordered by inclusion. A subset of $2^{[n]}$ is called a \emph{set system} or a (\emph{set}) \emph{family}; these terms are used interchangeably. A \emph{maximal chain} (also: \emph{full chain}) is
a saturated inclusion chain from $\varnothing$ to $[n]$, i.e., a sequence
\[
  \varnothing=A_0\subset A_1\subset\cdots\subset A_n=[n],\qquad \mathrm{~where~} |A_i|=i.
\]
There are exactly $n!$ such chains: each one is encoded by the
permutation that lists the elements in the order they are added. For a
family $\F\subseteq 2^{[n]}$, write $C(\F)$ for the number of
maximal chains \emph{entirely contained} in $\F$ (so each $A_i\in\F$),
and define the two normalized quantities
\[
  S(\F)=|\F|^{1/n},
  \qquad
  D(\F)=\left(\frac{C(\F)}{n!}\right)^{1/n}.
\]

Here $S(\F)\in[1,2]$ measures the exponential \emph{size} of $\F$
(per coordinate), while $D(\F)\in [0,1]$ is an exponential
measure of the chain-density of $\F$: large $D$ (close to one) means
$\F$ contains a large fraction of all $n!$ maximal chains. 
The basic
extremal quantity we study is
\[
  D_S
  \;=\;
  \sup\bigl\{\,D(\F):\F\subseteq 2^{[n]}\text{ for some }n,\;
  S(\F)\le S\,\bigr\},
  \qquad \mathrm{~where~} 1<S\le 2,
\]
or informally: ~\emph{How chain-dense can a family be if its size is at most
$S^n$?}

\smallskip

Of particular interest is the infimum of $S^2/D_S$, which controls the space-time product achievable for permutation-type problems (such as the TSP) in the framework of~\cite{ANW, DK}.

\smallskip

For convenience, we define the base-$2$ logarithms of the parameters:

$$S^{\ell}(\F) = {\log_2{S(\F)}}, \qquad D^{\ell}(\F) = {\log_2{(D(\F))}}, $$

and the extremal quantity

\[
  D^{\ell}_{S^{\ell}}
  \;=\;
  \sup\bigl\{\,D^{\ell}(\F):\F\subseteq 2^{[n]}\text{ for some }n,\;
  S^{\ell}(\F)\le S^{\ell}\,\bigr\},
  \qquad 0<S^{\ell}\le 1,
\]

noting that, for $\sigma\in(0,1]$:

$$D^\ell_\sigma = \log_2{(D_{2^\sigma})}.$$

\subsection{The entropy formulation}

Call a pair of random variables $(Z,T)$ \emph{admissible} if $T\sim\Unif[0,1]$ and $Z$ is an
arbitrary auxiliary random variable (not necessarily independent from $T$).  Denoting by $\HH$ the \emph{Shannon entropy} and by $\HH(\,\cdot\,|\,\cdot\,)$ the conditional entropy, define
\begin{equation*}
  h(Z,T)=\sup_{0\le t\le1}\HH(\1_{\{T\le t\}}\mid Z).
\end{equation*}

Informally, $h(Z,T)$ is the information content of $T$ being below a threshold $t$, given the realization of $Z$, for the threshold $t$ that maximizes this value. 

For a parameter $0<\sigma\le1$, we denote
\begin{equation*}
  \Dent_\sigma=-
  \inf\{\II(Z;T) \mid (Z,T)\text{ admissible and }h(Z,T)\le\sigma\}.
\end{equation*}

Here, $\II(Z;T)$ denotes the mutual information between $Z$ and $T$. Intuitively, there is a natural tension between $\sigma$ and $\Dent_\sigma$. If $h(Z,T)$ is suppressed to a low value $\sigma$, that is, if knowledge of $Z$ confers significant information about the outcome of $T$, then $Z$ and $T$ are mutually highly dependent; as a consequence, we expect $\Dent_\sigma$ to have a low value.  

The central result of our paper is that the tradeoff between $\sigma$ and $\Dent_\sigma$ is \emph{identical} to the tradeoff between the two combinatorial quantities described earlier. Specifically, we show the following.

\begin{theorem}\label{thm:comb_equals_entropy}
For every $0<\sigma\le1$,
\[
  D^{\ell}_{\sigma}=\Dent_\sigma.
\]
In particular, writing $\sigma = S^{\ell} = \log_2{S}$,
\begin{align*}
    \inf_{1<S\leq 2}\{\log_2 (S^2/D_S)\} &= \inf_{0<\sigma\leq 1}\{2\sigma - \Dent_\sigma\}\\
                     &= \inf_{\substack{0<\sigma\leq 1,\\(Z,T)}}\{2\sigma + \II(Z;T) \mid (Z,T)\text{ admissible and }h(Z,T)\le\sigma\}\\
                     &= \inf_{(Z,T)}\{2h(Z,T) + \II(Z;T) \mid (Z,T)\text{ admissible}\}.
\end{align*}
\end{theorem}

We will also derive upper and lower bounds for $\Dent_\sigma$, which, via the above theorem, entails bounds on the quantity $\gamma = S^2/D_S$ that appears as the base of the exponential in the TSP space-time tradeoff. Outside of this concrete algorithmic application, of course, other functions of $S$ and $D_S$ may be of interest as well, and these can, in principle, be optimized through similar methods.

We highlight the fact that the correspondence between the two settings is particularly clean: $h(Z,T)=\log_2 S$ and $\II(Z;T)=-\log_2D$.

\subsection{Concrete bounds}

We show how this entropy formulation of the problem makes it amenable to primal and dual certificates for rigorous computational upper and lower bounds respectively. In particular, we obtain the following.
\begin{theorem}\label{thm:prec}
    The best space-time product achievable for TSP (and other constant degree idempotent permutation problems) in the framework of{~\cite{ANW, DK}} is $\gamma^{n+o(n)}$, where $\gamma = \inf_{1<S\leq 2}\{(S^2/D_S)\}$ is such that
    \[3.1818 < \gamma < 3.18184.\]
\end{theorem}
The upper bound is given in Theorem~\ref{thm:gamma_ub}, and the lower bound in Theorem~\ref{thm:gamma_lb}.

We expect one can tighten the gap between the upper and lower bounds arbitrarily at the cost of more computation. We can also find bounds of the same nature on $D_S$ (or equivalently on $S^2/D_S$) as a function of $S$, yielding points close to the optimal space-time tradeoff curve achievable within the \cite{ANW, DK} framework. The upper and lower bound curves we obtain on $D_S$ are illustrated in Figure \ref{fig:density_bounds}. The implied bounds on $S^2/D_S$ (the relevant value for space-time tradeoffs) are illustrated in Figure \ref{fig:TS_bounds}. 

The upper bound curve on $D_S$ is obtained via lower bounds of the form $S^\Lambda/D_S \geq 2^{\mu_\Lambda}$ for different values of $\Lambda$. Each such bound directly gives $D_S \leq S^\Lambda/2^{\mu_\Lambda}$, and the illustrated upper bound curve of Figure \ref{fig:density_bounds} is the lower envelope of those individual curves. The lower bound curve on $D_S$ is obtained by finding certificates through optimization at discrete points, then interpolating between these points using the concavity of the function $\sigma \to \log_2 D_{2^\sigma}$. Both approaches lead to piecewise-defined curves which do not have a clean and short explicit closed-form expression. 

For the lower bound on $D_S$ (and thus the upper bound on $S^2/D_S$), we also derive a slightly weaker but more convenient closed-form expression as follows (see Figure~\ref{fig:density_closed_lb} for an illustration).

\begin{restatable}{theorem}{restateTCBAS}\label{thm:closed-bound-all-S}
For every $0<\sigma\leq 1$, let
\[
  k=\lfloor 1/\sigma\rfloor,
  \qquad \lambda=k(k+1)\sigma-k.
\]

Then $\lambda\in[0,1]$ and
\begin{equation*}
  \Dent_\sigma
  \ge
  -\lambda\log_2 k-(1-\lambda)\log_2(k+1)+I(1-\sigma),
\end{equation*}
where $I \geq 0.654865$.

Equivalently \textup{(}using the identity $\log_2 D_{2^\sigma} = \Dent_\sigma$\textup{)}, for $S=2^\sigma \in (1,2]$,
\begin{equation*}
  D_S \geq  k^{-\lambda}(k+1)^{\lambda-1}2^{I(1-\sigma)}.
\end{equation*}
\end{restatable}

\subsection{A further application}\label{sec14}

We also give a different application of our bounds on $\Dent_\sigma$, beyond the algorithmic improvements for TSP and other permutation problems.

Consider the Boolean lattice $2^{[n]}$ where two subsets of $[n]$ are \emph{comparable} if one includes the other. An \emph{antichain} is a collection of pairwise incomparable subsets of $[n]$. An antichain is maximal if it cannot be extended by adding one more set. A \emph{fibre} is a subset of $2^{[n]}$ that intersects every maximal antichain of $2^{[n]}$.

Lonc and Rival~\cite{LoncRival1987} raised the question of the minimal size of a fibre in $2^{[n]}$.
Denoting this quantity by $\mathrm{MF}_n$, we observe $\mathrm{MF}_n \leq \fO(2^{n/2})$, since the family of all sets comparable to any given set of the middle layer is a fibre. Lonc and Rival~\cite{LoncRival1987} conjectured this upper bound to be tight.  Duffus, Sands and Winkler~\cite{DuffusSW90} made some progress towards the conjecture, showing $\mathrm{MF}_n \geq \Omega(1.25^{n})$, which was later improved by {\L}uczak to $\mathrm{MF}_n \geq \Omega(2^{n/3})$, as reported by Duffus and Sands~\cite{DuffusSands2001}. There has been no further progress since.

As a corollary of our results we obtain the following.

\begin{restatable}{theorem}{restateFIBRE}\label{thm:fibre}
    Let $\sigma\in (0,1]$. If $D^{\ell}_\sigma \leq -1$, then $\mathrm{MF}_n \geq 2^{\sigma n}$.

    In particular, since $D^{\ell}_{0.37586} \leq -1$, we have
    \[\mathrm{MF}_n \geq 2^{0.37586 n}.\]
\end{restatable}

\newpage

\begin{figure}
    \centering
    \includegraphics[width=0.85\linewidth]{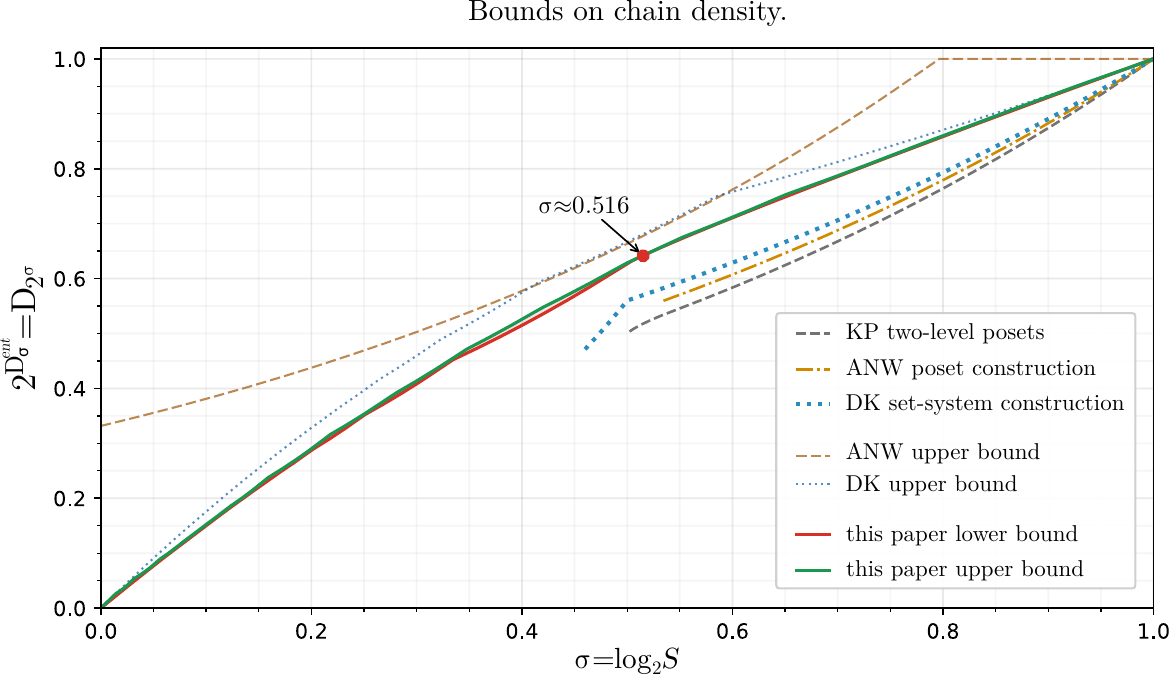}
    \caption{Upper and lower bound curves for the tradeoff between (logarithmic) set system size $\sigma = \log_2{S} = S^{\ell}$ (horizontal axis) and set system density $D_S = D_{2^{\sigma}}$ (vertical axis). Red dot indicates the best point-wise tradeoff, i.e., the minimal $S^2 / D_S$ value. Solid red and green indicate the lower and upper bound curves we derive. Note that these are piecewise-defined curves resulting from certificates computed at discrete points; while the bounds could be further refined through more computation, a compact, closed-form description of the optimal tradeoff curve is not available. Upper and lower bounds found in previous works \cite{KoivistoParviainen2010,ANW, DK} are also shown. }
    \label{fig:density_bounds}
\end{figure}

\vspace{-0.9cm}
\enlargethispage{3\baselineskip}

\begin{figure}[H]
    \centering
    \includegraphics[width=0.85\linewidth]{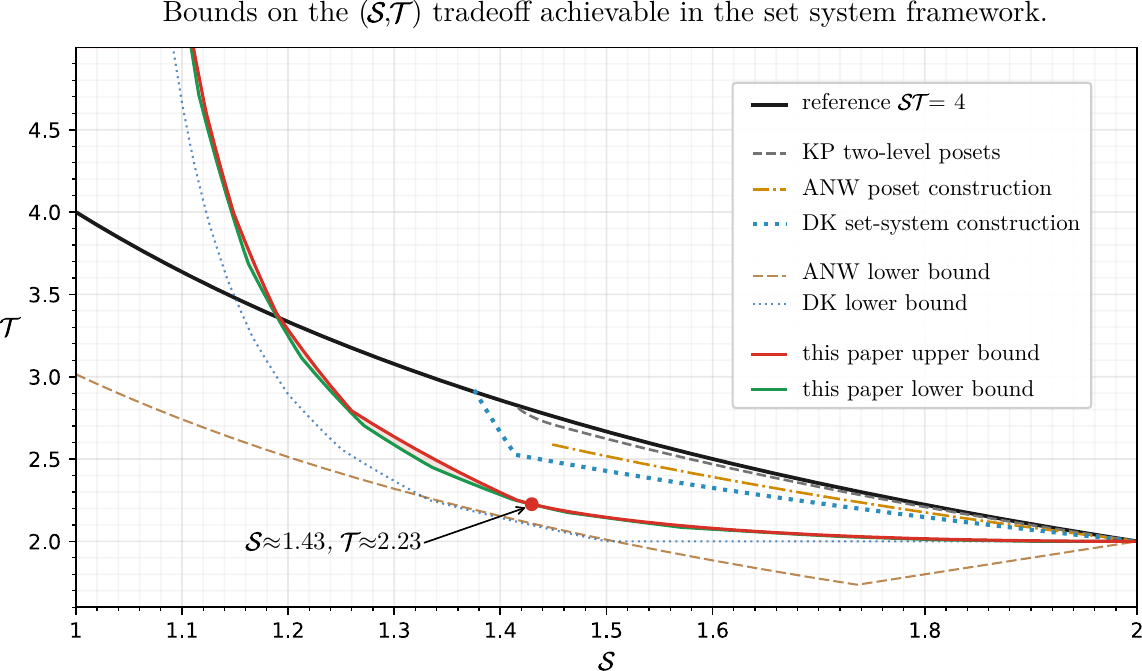}
    \caption{Tradeoff curves ($\cT$ vs.~$\cS$) for the best exponential running time $\cT^n$ and space $\cS^n$ achievable for the TSP, with optimal space-time product shown as a red point. Classical $\cS \cdot \cT = 4$ tradeoff shown for reference as black line. It may seem that the method improves this baseline only for $\cS \gtrsim 1.19$; note however, that the improvements can be extended to the lower range of $\cS$, by combining the method with Gurevich-Shelah divide-and-conquer. For more details on this, by-now-standard technique, refer to~\cite[Lemma~3]{Quantum}, \cite[Lemma~3.4]{DK}.  }
    \label{fig:TS_bounds}
\end{figure}

\paragraph{Structure of the paper.} In \S\,\ref{sec:prelim} we introduce the basic information theory toolkit we need. In \S\,\ref{sec3} we prove the two-way correspondence between the combinatorial and information-theoretic formulations of the problem, proving Theorem~\ref{thm:comb_equals_entropy}. Sections \S\,\ref{sec4}, \S\,\ref{sec5} are dedicated to proving lower bounds for $\Dent_\sigma$, yielding the algorithmic side of Theorem~\ref{thm:prec} and Theorem~\ref{thm:closed-bound-all-S}. Section \S\,\ref{sec6} contains the upper bounds for $\Dent_\sigma$, yielding the barrier side of Theorem~\ref{thm:prec}. 
In \S\,\ref{sec8} we describe the application to bounding minimal fibres in the Boolean lattice, proving Theorem~\ref{thm:fibre}. In \S\,\ref{sec:concl} we conclude with some open questions.

\paragraph{Disclosure on the use of AI. } The results of the paper were obtained with the help of OpenAI's ChatGPT 5.5 tool, which was also used to verify and reformulate arguments and to explore the problem structure, which materially affected all technical sections of the paper. The authors assume full responsibility for the contents, including its correctness and originality. 

\section{Preliminaries}\label{sec:prelim}

Recall the definition of $D_S$ as a supremum of the density $D(\cF)$ over all families with finite ground set, with $S(\cF) \leq S$. For $\alpha\le 0$, a \emph{witness} for a lower bound $\upalpha \leq D_S^\ell$, or simply a \emph{set system witness}, is a family $\cF$ with $S(\cF) \leq S$ and $D^\ell(\cF) \geq \upalpha$. Note that such a witness also implies a corresponding upper bound on $\gamma$.

Similarly, a witness for a lower bound $\upalpha \leq \Dent_\sigma$ is a pair $(Z,T)$, where $T\sim\Unif[0,1]$ and $Z$ is an arbitrary random variable, where $h(Z,T) \leq \sigma$, and $-\II(Z;T) \geq \upalpha$. We refer to it as a \emph{$(Z,T)$ witness}.

\paragraph{Information theory. }

In the following, we briefly list some properties and facts from information theory
that we use; all of them are standard, e.g., see~\cite{cover_book}.

We refer to~\cite{cover_book} for the standard definitions of the entropy $\HH(X)$ of a random variable $X$ or $\HH(p)$ of a probability vector $p$, and the conditional entropy $\HH(X \mid Y)$, as well as for the definition of the mutual information $\II(X;Y)$ of two random variables, and the conditional mutual information $\II(X;Y \mid W)$.

An easy upper bound on the entropy is $
  \HH(X)\le\log_2|\mathcal{X}|$,
where $\mathcal{X}$ is the support set of $X$, and applying a function cannot increase entropy: $$\HH(f(X))\le\HH(X).$$

Since conditioning cannot increase entropy,
\[
  \HH(X\mid Y,W)\le\HH(X\mid Y);
\]
in particular, conditioning on a function of $Y$ instead of $Y$ can only
increase entropy: $$\HH(X\mid f(Y))\ge\HH(X\mid Y).$$

Finally, $\HH(X\mid Y)=0$ if and only if $X$ is a function of $Y$ (up to a
set of probability zero).

When $X$ is discrete,
\[
  \II(X;Y)=\HH(X)-\HH(X\mid Y)\;\le\;\HH(X);
\]
in particular $\II(X;Y)=\HH(X)$ when $X$ is a function of $Y$.
For the continuous entropy, if $Y$ has a density and the
quantities below are finite, then the analogous statement holds:
\[
  \II(X;Y)=\Hc(Y)-\Hc(Y\mid X).
\]

The chain rule for mutual information states
\[
  \II(X,W;Y)=\II(W;Y)+\II(X;Y\mid W).
\]
Two consequences we use: if $R$ is independent of the pair $(X,Y)$, then
$\II(X;(Y,R))=\II(X;Y)$; and if $B$ is independent of $Y$, then
$\II((B,X);Y)=\II(X;Y\mid B)$.

The \emph{data processing inequality} says that for every
(measurable) function $f$,
\[
  \II(X;f(Y))\le\II(X;Y).
\]
In words, processing a variable cannot increase mutual information. 
By symmetry the same holds for a function applied to $X$.

\section{Equivalence between the combinatorial and entropy formulations}\label{sec3}

In this section we show the main equivalence between the combinatorial and information-theoretic formulations of our tradeoff, proving Theorem~\ref{thm:comb_equals_entropy}. 

\subsection{From set system witnesses to $(Z,T)$ witnesses}
Here we show that any set system witnessing 

$D^{\ell}_{\sigma} \geq \alpha$ can be used to construct an 
admissible pair $(Z,T)$ witnessing $\Dent_\sigma \geq \alpha$, and thus $D^{\ell}_{\sigma} \leq \Dent_\sigma$.

\begin{lemma}\label{lem:set-to-entropy-witness}
Let $\F\subseteq2^{[n]}$ satisfy $C(\F)>0$.  Then there is an admissible pair
$(Z,T)$ such that
\[
  h(Z,T)\le \frac1n\log_2|\F|=S^{\ell}(\cF),
  \qquad
  \II(Z;T)=\frac1n\log_2\frac{n!}{C(\F)}=- D^{\ell}(\F).
\]
\end{lemma}

\begin{proof}
Let $U_1,\ldots,U_n$ be independent uniform random variables on $[0,1]$.
Their increasing order $\pi$ determines a uniformly random maximal chain: if
$U_{\pi(1)}<\cdots<U_{\pi(n)}$, then the chain is
\[
  \varnothing,
  \{\pi(1)\},
  \{\pi(1),\pi(2)\},
  \ldots,
  [n].
\]
Let $E_\F$ be the event that this chain is contained in $\F$, and let
$\Qbb=\Pbb(\,\cdot\mid E_\F)$.  Since each of the $n!$ maximal chains is
equally likely,
\[
  \Pbb(E_\F)=\frac{C(\F)}{n!}.
\]

Choose $J$ uniformly from $[n]$, independently of the $U_i$'s under $\Qbb$, and
set
\[
  T=U_J,
  \qquad
  Z=(J,U_1,\ldots,U_{J-1}).
\]
The event $E_\F$ depends only on the order of the $U_i$'s, not on their sorted
values.  Therefore a uniformly chosen coordinate remains uniform after
conditioning on $E_\F$ and $T\sim\Unif[0,1]$ under $\Qbb$.

Under the original probability distribution, 
$(U_1,\ldots,U_n)$ is uniform on the cube $[0,1]^n$.  The event $E_\F$ is a
union of exactly $C(\F)$ disjoint regions of $[0,1]^n$, one for each valid maximal chain.  Each
order region has volume $1/n!$.  Therefore, under $\Qbb$, the vector
\[
  U=(U_1,\ldots,U_n)
\]
is uniform on a region of volume
\[
  \frac{C(\F)}{n!}.
\]
Consequently
\[
  \Hc(U_1,\ldots,U_n)=\log_2\frac{C(\F)}{n!} = D^{\ell}(\F).
\]
By the chain rule for continuous entropy,
\[
  \Hc(U_1,\ldots,U_n)
  =\sum_{j=1}^n\Hc(U_j\mid U_1,\ldots,U_{j-1}).
\]

Since $T=U_J$ and $Z=(J,U_1,\ldots,U_{J-1})$, we have
\begin{align*}
    \Hc(T \mid Z)   &= \Hc(U_J \mid J, U_1, \ldots, U_{J-1})\\
                    &= \sum_{j=1}^n \Qbb(J=j) \cdot \Hc(U_j | U_1, \ldots, U_{j-1})\\
                    &= \frac{1}{n}\sum_{j=1}^n \Hc(U_j | U_1, \ldots, U_{j-1})\\ 
                    &= \frac{1}{n}\Hc(U_1,\ldots,U_n).
\end{align*}

On the other hand, $T$ is uniform on $[0,1]$, so $\Hc(T)=\log_2 1=0$.  Thus
\[
\begin{aligned}
  \II(Z;T)
  &=\Hc(T)-\Hc(T\mid Z) \\
  &=-\frac1n\Hc(U_1,\ldots,U_n) \\
  &= -D^{\ell}(\cF).
\end{aligned}
\]

It remains to bound the threshold entropy.  Fix $t$ and define
$B_j(t)=\1_{\{U_j\le t\}}$ and $X_t=(B_1(t),\ldots,B_n(t))$. Interpret the random vector $X_t$ as a (random) subset of $[n]$ (a value of $1$ at some index means this index is included in the subset).  Note that this subset is a prefix set of the maximal chain determined by $U_1,\ldots,U_n$, and thus is in $\cF$ if the latter is a maximal chain contained in $\cF$. Under $\Qbb$ this is true with probability $1$, so the subset of $[n]$ represented by $X_t$ is in $\cF$, and $X_t$ has support at most $|\cF|$. Therefore
\[
  \HH(X_t)\le\log_2|\F|.
\]
Moreover $\1_{\{T\le t\}}=B_J(t)$, so
\begin{align*}
  \HH(\1_{\{T\le t\}}\mid Z)
  &=\frac1n\sum_{j=1}^n \HH(B_j(t)\mid U_{<j}) \\
  &\le \frac1n\sum_{j=1}^n
       \HH(B_j(t)\mid B_1(t),\ldots,B_{j-1}(t)) \\
  &=\frac1n\HH(X_t)
   \le \frac1n\log_2|\F| = S^{\ell}(\F).
\end{align*}

Taking the supremum over $t$ proves the lemma.
\end{proof}

\input{section_3_2_arx}

\section{Algorithm from Theorem~\ref{thm:prec} via a lower bound on $\Dent_\sigma$, for $\frac{1}{2}\leq \sigma\leq 1$}\label{sec4}
We now describe explicit witnesses for lower bounds on $\Dent_\sigma$ for $\sigma\geq \frac{1}{2}$. By Theorem~\ref{thm:comb_equals_entropy} and the framework from~\cite{ANW,DK}, this entails the algorithmic side of Theorem~\ref{thm:prec} (upper bound on $\gamma$).

The value of the resulting bounds will be expressed in terms of an integral which we now define.
For $\sigma\in(1/2,1)$, let $(x_\sigma,y_\sigma)\in(1/2,1)\times(0,1/2)$ be
defined by the equations
\begin{equation}\label{eq:tangent-point}
  \hbin(x_\sigma)+\hbin(y_\sigma)=2\sigma,
  \qquad
  x_\sigma\hbin'(x_\sigma)+y_\sigma\hbin'(y_\sigma)=0.
\end{equation}
Let $f_\sigma$ be the lower branch of the level curve
\begin{equation*}
  \hbin(x)+\hbin(f_\sigma(x))=2\sigma,
\end{equation*}
for $x_\sigma\le x\le1-y_\sigma$.  Define
\begin{equation*}
  I_\sigma=
  \int_{x_\sigma}^{1-y_\sigma}
  \Phi(f_\sigma'(x))\,dx,
  \qquad
  \Phi(r)=(1+r)\hbin\left(\frac1{1+r}\right),
\end{equation*}
and
\begin{equation*}
  a_\sigma=
  (x_\sigma+y_\sigma)
  \hbin\left(\frac{x_\sigma}{x_\sigma+y_\sigma}\right)
  +\frac12 I_\sigma.
\end{equation*}
At the endpoint $\sigma=1/2$, we interpret $(x_{1/2},y_{1/2})=(1/2,0)$ and
write
\begin{equation*}
  I=I_{1/2},
\end{equation*}
so that $a_{1/2}=I/2$.

\begin{figure}
    \centering
    \includegraphics[width=0.5\linewidth]{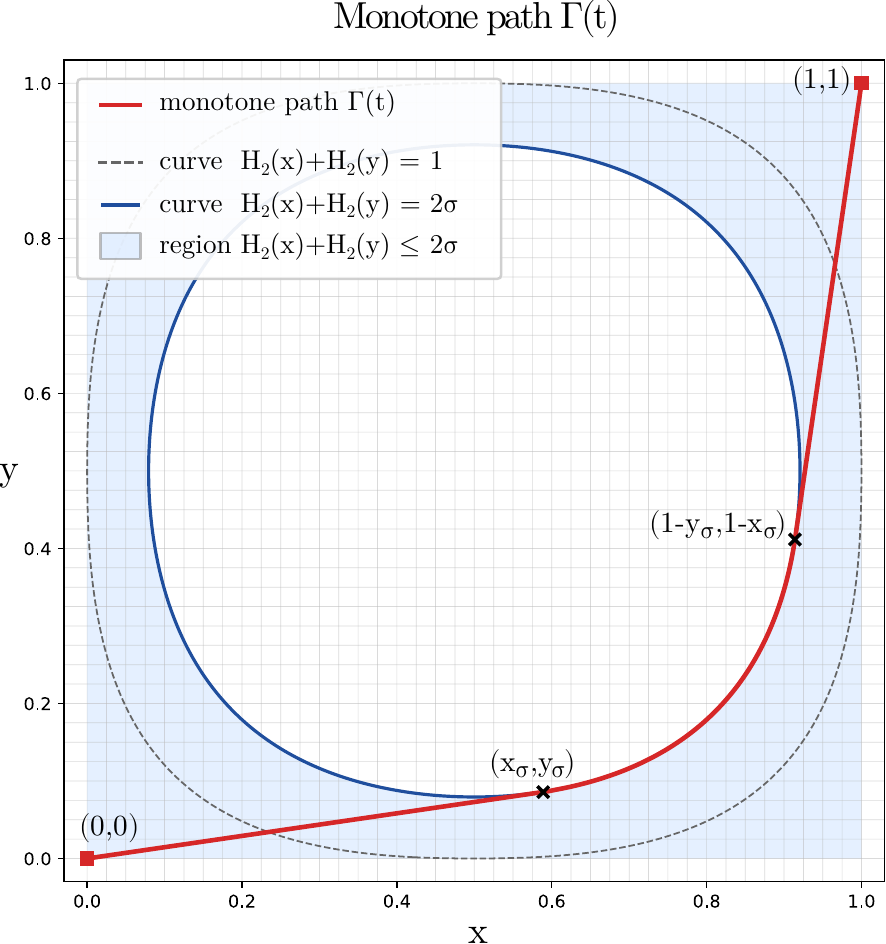}
    \caption{Illustration of the monotone path $\Gamma(t)=(F_1(t),F_2(t))$ used in the proof of Theorem~\ref{thm:two-state-corridor-close}. The path stays inside the region $\hbin(x)+\hbin(y)\leq 2\sigma$ at all times, hugging the boundary $\hbin(x)+\hbin(y) = 2\sigma$ between $(x_\sigma,y_\sigma)$ and $(1-y_\sigma,1-x_\sigma)$.}
    \label{fig:monotone_path}
\end{figure}

\begin{theorem}\label{thm:two-state-corridor-close}
For every $\sigma\in[1/2,1)$,
\[
  \Dent_{\sigma} \geq a_\sigma-1.
\]
\end{theorem}

\begin{proof}
Take $Z$ uniform on $\{1,2\}$.  We specify the conditional CDFs
$F_1,F_2$ of $T$ given $Z=1,2$ by a monotone path
\[
  \Gamma(t)=(F_1(t),F_2(t))
\]
from $(0,0)$ to $(1,1)$ satisfying $(F_1(t)+F_2(t))/2=t$.  

For any fixed $t\in[0,1]$, we have
\begin{align*}
    \HH(\1_{\{T\leq t\}} \mid Z) 
    &= \frac{1}{2}\HH(\1_{\{T\leq t\}} \mid Z=1) + \frac{1}{2}\HH(\1_{\{T\leq t\}} \mid Z=2)\\
    &= \frac{1}{2}\hbin(F_1(t)) + \frac{1}{2}\hbin(F_2(t)).
\end{align*}
Thus, the quantity
\begin{align*}
    h(Z,T) = \sup_{0\leq t \leq 1} \HH(\1_{\{T\leq t\}} \mid Z)
\end{align*}
is at most $\sigma$ whenever the path $\Gamma(t)=(F_1(t),F_2(t))$ stays in the region $\hbin(x)+\hbin(y)\leq 2\sigma$ for all $t\in[0,1]$.

For $\sigma\in [1/2,1)$, use the path consisting of the line segment from
$(0,0)$ to $(x_\sigma,y_\sigma)$, the boundary arc
$y=f_\sigma(x)$ from $(x_\sigma,y_\sigma)$ to
$(1-y_\sigma,1-x_\sigma)$, and the reflected line segment to $(1,1)$. See Figure \ref{fig:monotone_path} for an illustration. The
tangency condition in \eqref{eq:tangent-point} ensures that the line segment
lies inside the entropy sublevel region; the arc lies on its boundary, and the
last segment follows by symmetry.  Hence this is a valid witness with
$h(Z,T)\le\sigma$.

It remains to compute $\II(Z;T) = \HH(Z)-\HH(Z \mid T)$. 
We have $\HH(Z)=1$, since $Z$ takes one of two values with equal probability.

Fix $t$, and assume that for an infinitesimal increment from $t$ to $t+dt$, $F_1(t)$ increases by $dx$ and $F_2(t)$ by $dy$. Then we have
\begin{align*}
    \Pbb(t\leq T \leq t+dt) 
    &= \frac{1}{2}\Pbb(t\leq T \leq t+dt \mid Z=1) + \frac{1}{2}\Pbb(t\leq T \leq t+dt \mid Z=2) \\
    &= \frac{1}{2}(dx+dy).
\end{align*}
By Bayes’ rule, 
\begin{align*}
    \Pbb(Z=1 \mid t\leq T \leq t+dt) &= \frac{dx}{dx+dy}.
\end{align*}
Hence, the contribution to $\HH(Z \mid T)$ of this infinitesimal piece of the path is 
\[
\Pbb(t\leq T \leq t+dt) \cdot \HH(Z \mid t\leq T \leq t+dt) 
= \frac{1}{2}(dx+dy)\cdot\hbin\left(\frac{dx}{dx+dy}\right).
\]
The boundary arc therefore contributes exactly $\frac{1}{2}I_\sigma$.  The first straight
segment contributes
\[
  \frac{1}{2}(x_\sigma+y_\sigma)
  \hbin\left(\frac{x_\sigma}{x_\sigma+y_\sigma}\right),
\]
and the reflected segment contributes the same amount.

Thus we have $\HH(Z \mid T) = a_\sigma$, and $\II(Z;T) = 1-a_\sigma$.
\end{proof}

For any given $\frac{1}{2} \leq \sigma < 1$ we can compute (lower bounds on) $I_\sigma$ and $a_\sigma$ numerically up to a specified precision using standard numerical integration methods, to get concrete lower bounds on $\Dent_\sigma$. The interval over which the integral is computed is discretized into many subintervals, and the contribution of each subinterval is lower bounded using the concavity of the binary entropy function. 

For example, we have the following.
\begin{lemma}\label{lem:a_sigma_lb}
    The following lower bounds hold.
    \begin{align*}
        &a_{0.5} \geq 0.3274327, \quad
        a_{0.6} \geq 0.5073334, \quad
        a_{0.7} \geq 0.6522079, \\
        &a_{0.8} \geq 0.7800100, \quad
        a_{0.9} \geq 0.8952229.
    \end{align*}
\end{lemma}

In particular, for $\sigma = \frac{1}{2}$ we get
\begin{lemma}
    \[2a_\frac{1}{2}=I=I_\frac{1}{2} \geq 0.654865\]
    Thus,
    \[\Dent_\frac{1}{2} \geq a_\frac{1}{2} -1 \geq -0.672568.\]
\end{lemma}

Optimizing $2\sigma-(a_\sigma-1)$ we get that this value is minimized roughly at $\sigma \approx 0.515617$, where we get the following.
\begin{theorem}\label{thm:gamma_ub}
    Let $\sigma^* = 0.515617$. We have
    \[a_{\sigma^*} \geq 0.361376.\]
    Thus,
    \[\Dent_{\sigma^*} \geq a_{\sigma^*} - 1 \geq -0.638624,\]
    and
    \begin{align*}
        \inf_{0<\sigma\leq 1}\{2\sigma-\Dent_\sigma\} \leq 2\sigma^*-\Dent_{\sigma^*} \leq  1.669858
    \end{align*}
    In terms of $S$ and $D_S$, this is
    \[
      \inf_{1<S \leq 2}{\frac{S^2}{D_S}} \leq 2^{1.669858} < 3.18184.
    \]
\end{theorem}

\section{Closed-form lower bound on $\Dent_\sigma$, for $0<\sigma\leq 1$}\label{sec5}

\begin{figure}
    \centering
    \includegraphics[width=0.8\linewidth]{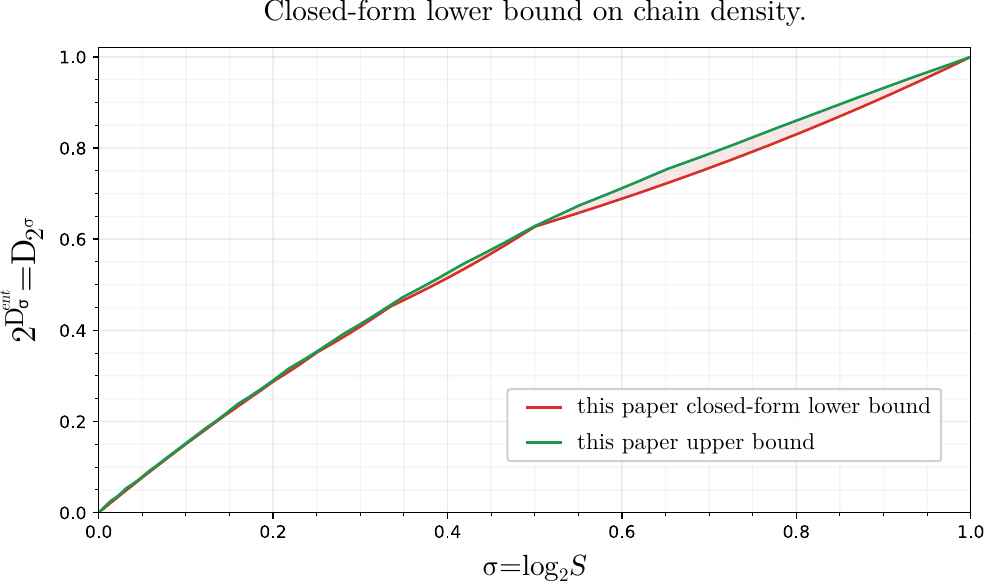}
    \caption{Closed-form lower bound curve for chain density $D_S$ (Theorem~\ref{thm:closed-bound-all-S}) and obtained upper bound curve.}
    \label{fig:density_closed_lb}
\end{figure}

In this section we adapt the $\sigma=\frac{1}{2}$ witness from the previous section into a
closed-form bound for every value of $0<\sigma\leq 1$, yielding Theorem~\ref{thm:closed-bound-all-S}. Recall that the preceding section gives
\[
  \Dent_{1/2}\geq \frac{1}{2}I-1.
\]
The idea is to use the same two-state transition repeatedly to get witnesses at
$\sigma=1/k$, and then interpolate between these points using concavity of
$\Dent_\sigma$.

\restateTCBAS*

Note that although this bound is not tight in general, it has the advantage of having a simple expression, while not being too far from the computed upper bounds (as illustrated in Figure \ref{fig:density_closed_lb}).

We first prove the discrete estimates at the points $\sigma=1/k$.

\begin{lemma}\label{lem:gridpoint-Dent-bound}
For every integer $k\ge1$,
\begin{equation*}
  \Dent_{1/k}\ge -\log_2 k+\frac{k-1}{k}I.
\end{equation*}
\end{lemma}

\begin{proof}
The case $k=1$ is witnessed by taking $Z$ to be constant, so assume
$k\ge2$.  Take $Z$ uniform on $[k]=\{1,\ldots,k\}$.  We specify the
conditional CDFs
\[
  F_i(t)=\Pbb(T\le t\mid Z=i),\qquad i=1,\ldots,k,
\]
by a monotone path
\[
  \Gamma(t)=(F_1(t),\ldots,F_k(t))
\]
from $(0,\ldots,0)$ to $(1,\ldots,1)$ satisfying
\[
  \frac1k\sum_{i=1}^k F_i(t)=t.
\]
This last identity guarantees that the unconditional distribution of $T$ is
uniform on $[0,1]$.

The path is as follows.  First move the first coordinate from $0$ to $1/2$.
Then, for $j=1,\ldots,k-1$, move only the pair of coordinates
$(F_j,F_{j+1})$ from $(1/2,0)$ to
$(1,1/2)$ along the arc $\hbin(F_j)+\hbin(F_{j+1}) = 1$. Finally move the last coordinate from $1/2$ to $1$.  All other
coordinates remain fixed during each step. See Figure \ref{fig:monotone_path_3D} for an illustration.

At every point of this path, the sum of the coordinate entropies is at most
$1$. Therefore, for every threshold $t$,
\[
  \HH(\1_{\{T\le t\}}\mid Z)
  =\frac1k\sum_{i=1}^k \hbin(F_i(t))
  \le \frac1k.
\]
So this is a valid witness for $\Dent_{1/k}$.

It remains to compute its mutual information.  We use
\[
  \II(Z;T)=\HH(Z)-\HH(Z\mid T).
\]
Since $Z$ is uniform on $[k]$, we have $\HH(Z)=\log_2 k$.

Consider an infinitesimal piece of the CDF path.  Suppose the coordinates
increase by
\[
  (dx_1,\ldots,dx_k),
\]
where only some of these increments may be nonzero.  Conditional on $Z=i$, the
probability that $T$ falls in this infinitesimal piece is $dx_i$.  Since the
states of $Z$ are equally likely, the unconditional probability of $T$ falling in this piece is
\[
  \frac1k\sum_i dx_i.
\]
Given that $T$ lies in this piece, Bayes' rule says that the posterior
probability of $Z=i$ is
\[
  \frac{dx_i}{\sum_r dx_r}.
\]
Thus the contribution of this piece to $\HH(Z\mid T)$ is
\[
  \frac1k\left(\sum_i dx_i\right)
  \HH\left(\frac{dx_1}{\sum_r dx_r},\ldots,
            \frac{dx_k}{\sum_r dx_r}\right).
\]
On a straight piece where only one coordinate moves, this contribution is $0$.
On each two-coordinate endpoint arc, the total contribution is exactly $I/k$,
because the corresponding two-coordinate path entropy is the integral $I$.
There are $k-1$ such arcs.  Hence
\[
  \HH(Z\mid T)=\frac{k-1}{k}I.
\]
Consequently
\[
  \II(Z;T)=\log_2 k-\frac{k-1}{k}I.
\]
Since $\Dent_{1/k}$ is the negative of the least possible mutual information
among witnesses with $h(Z,T)\le1/k$, this gives
\[
  \Dent_{1/k}\ge
  -\log_2 k+\frac{k-1}{k}I,
\]
as claimed.
\end{proof}

\begin{figure}
    \centering
    \includegraphics[width=0.5\linewidth]{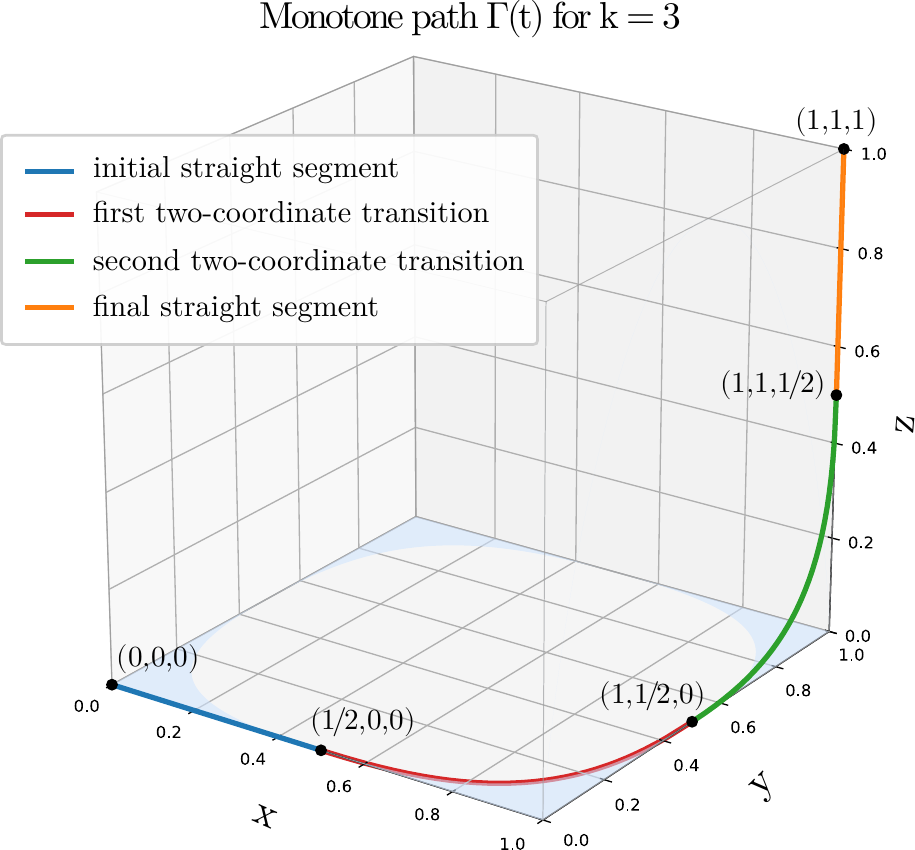}
    \caption{Illustration of the monotone path $\Gamma(t)=(F_1(t),F_2(t),F_3(t))$ used in the proof of Theorem \ref{thm:closed-bound-all-S} (in the $k=3$ case). The path stays inside the region $\hbin(x)+\hbin(y)+\hbin(z)\leq 1$ (whose intersection with the lower faces of the cube is shown in light blue) at all times. It is composed of four pieces: an initial straight segment, two distinct two-coordinate transitions staying on a face of the cube, and a final straight segment. The two coordinate transitions trace the same path on their respective faces as the pieces of the curve in Figure \ref{fig:monotone_path} between $(\frac{1}{2},0)$ and $(1,\frac{1}{2})$ in the $\sigma=\frac{1}{2}$ case.}
    \label{fig:monotone_path_3D}
\end{figure}

We also need the following basic properties of $\Dent_\sigma$.

\begin{lemma}\label{lem:Dent-concave}
The function $\sigma\mapsto\Dent_\sigma$ is finite, continuous, nondecreasing,
and concave on $(0,1]$, with $\Dent_1=0$.
\end{lemma}

\begin{proof}
The function $\Dent$ is nondecreasing in $\sigma$, because increasing $\sigma$
relaxes the constraint $h(Z,T)\le\sigma$.

To prove concavity, let $0<\sigma_1,\sigma_2\le1$ and $0\le\lambda\le1$.
Choose admissible pairs $(Z_1,T)$ and $(Z_2,T)$ with
$h(Z_i,T)\le\sigma_i$.  We may use the same uniform variable $T$ for both
witnesses by describing each witness through the conditional law of $Z_i$ given
$T=t$.

Let $B$ be independent of $T$, with
\[
  \Pbb(B=1)=\lambda,
  \qquad
  \Pbb(B=2)=1-\lambda.
\]
Conditional on $B=i$ and $T=t$, sample $Z_i$ according to the conditional law
of the $i$th witness.  Set $W=(B,Z_B)$.  Then
\[
  \II(W;T)=\lambda\II(Z_1;T)+(1-\lambda)\II(Z_2;T),
\]
and, for every threshold $t$,
\[
  \HH(\1_{\{T\le t\}}\mid W)
  =\lambda\HH(\1_{\{T\le t\}}\mid Z_1)
   +(1-\lambda)\HH(\1_{\{T\le t\}}\mid Z_2).
\]
Therefore
\[
  h(W,T)\le\lambda\sigma_1+(1-\lambda)\sigma_2.
\]
Taking infima over the two witnesses shows that $-\Dent_\sigma$ is convex, and thus $\Dent_\sigma$ is concave.

For finiteness, fix $\sigma>0$ and choose an integer $K$ with $1/K\le\sigma$.
Divide $[0,1]$ into $K$ equal bins
\[
  B_j=\left(\frac{j-1}{K},\frac jK\right],\qquad j=1,\ldots,K,
\]
and let $Z$ be the bin containing $T$.  Given $Z$, the threshold bit
$\1_{\{T\le t\}}$ is uncertain only when $t$ lies in the same bin as $T$;
this occurs with probability $1/K$.  Hence
\[
  h(Z,T)\le \frac1K\le\sigma.
\]
Also $\II(Z;T)=\HH(Z)=\log_2K<\infty$.  Thus $\Dent_\sigma> -\infty$.
Since mutual information is always nonnegative, we also have
$\Dent_\sigma\le0$, so $\Dent_\sigma$ is finite.

Finally, $\Dent_1=0$: the constant witness gives $\Dent_1\ge0$, while
nonnegativity of mutual information gives $\Dent_1\le0$.  A finite concave
function is continuous on the interior of its domain, and continuity at $1$
follows from concavity together with the value $\Dent_1=0$.
\end{proof}

\begin{proof}[Proof of \cref{thm:closed-bound-all-S}]
If $\sigma = 1$, both sides of the claimed inequality are $0$.  So assume
$0<\sigma<1$.

With $k=\lfloor1/\sigma\rfloor$, we have
\[
  \frac1{k+1}<\sigma\le\frac1k.
\]
The number
\[
  \lambda=k(k+1)\sigma-k
\]
is exactly the coefficient for writing $\sigma$ as a convex combination of the
neighboring grid points:
\[
  \sigma=\lambda\frac1k+(1-\lambda)\frac1{k+1}.
\]
By concavity of $\Dent$ and \cref{lem:gridpoint-Dent-bound},
\begin{align*}
  \Dent_\sigma
  &\ge
  \lambda\Dent_{1/k}+(1-\lambda)\Dent_{1/(k+1)} \\
  &\ge
  \lambda\left(-\log_2 k+\frac{k-1}{k}I\right)
  +(1-\lambda)\left(-\log_2(k+1)+\frac{k}{k+1}I\right) \\
  &=-\lambda\log_2 k-(1-\lambda)\log_2(k+1)
    +I\left(\lambda\frac{k-1}{k}+(1-\lambda)\frac{k}{k+1}\right).
\end{align*}
The coefficient of $I$ simplifies to
\[
  \lambda\frac{k-1}{k}+(1-\lambda)\frac{k}{k+1}
  =1-\left(\lambda\frac1k+(1-\lambda)\frac1{k+1}\right)
  =1-\sigma.
\]
This proves the theorem.
\end{proof}

\section{Computing a lower bound on $2\sigma-\Dent_\sigma$, and upper bounds on $\Dent_\sigma$}\label{sec6}
In this section we explain how the lower bound on
\[
  \gamma
  :=
  \inf_{0<\sigma\le1}\{2\sigma-\Dent_\sigma\}
  =
  \inf_{(Z,T)}\{\II(Z;T)+2h(Z,T)\}
\]
is computed.  We note that this entails the lower bound of Theorem~\ref{thm:prec} (lower bound on $\gamma$), again by using the equivalence Theorem~\ref{thm:comb_equals_entropy} and the algorithmic framework.

The right-hand side is an optimization over all admissible
pairs $(Z,T)$, and is therefore infinite-dimensional.  The key point is that
the finite discretization below approximates the continuous problem with an
explicit error bound.

Fix an integer $K\ge2$, let $Y\sim\Unif([K])$, and define
\begin{equation*}
  \mathsf B_K(Z,Y)
  =
  \II(Z;Y)
  +2\max_{1\le k<K}\HH(\1_{\{Y\le k\}}\mid Z).
\end{equation*}
Let
\begin{equation*}
  \gamma_K
  =
  \inf_{Y\sim\Unif([K])}\mathsf B_K(Z,Y),
\end{equation*}
where the infimum is over all couplings of the auxiliary variable $Z$ with the
uniform variable $Y$.

\begin{lemma}\label{lem:finite-dual-transfer}
For every $K\ge2$,
\begin{equation*}
  \gamma_K\le \gamma\le \gamma_K+2\hbin(1/K).
\end{equation*}
\end{lemma}

\begin{proof}
We first prove the lower bound $\gamma\ge \gamma_K$.  Let $(Z,T)$ be any
admissible pair, and discretize $T$ into $K$ equal bins:
\[
  Y=\min(K,\lceil KT\rceil),
\]
with the convention that $Y=1$ when $T=0$.  Since $T$ is uniform on
$[0,1]$, the variable $Y$ is uniform on $[K]$.

The variable $Y$ is a function of $T$, so by data processing,
\[
  \II(Z;Y)\le \II(Z;T).
\]
Also, for each $1\le k<K$,
\[
  \1_{\{Y\le k\}}=\1_{\{T\le k/K\}},
\]
and therefore
\[
  \HH(\1_{\{Y\le k\}}\mid Z)
  \le
  \sup_{0\le t\le1}\HH(\1_{\{T\le t\}}\mid Z)
  =h(Z,T).
\]
Thus
\[
  \mathsf B_K(Z,Y)
  \le
  \II(Z;T)+2h(Z,T).
\]
Since $\gamma_K$ is the infimum of $\mathsf B_K$ over all finite $K$-bin witnesses,
$\mathsf B_K(Z,Y)\ge \gamma_K$.  Hence
\[
  \II(Z;T)+2h(Z,T)\ge \gamma_K.
\]
Taking the infimum over all admissible $(Z,T)$ gives $\gamma\ge \gamma_K$.

We now prove the upper bound $\gamma\le \gamma_K+2\hbin(1/K)$.  Let $(Z,Y)$ be a
witness for the finite version with $Y\sim\Unif([K])$.  Turn it into a continuous
witness by spreading each bin uniformly: let $R\sim\Unif[0,1]$ be independent
of $(Z,Y)$ and put
\[
  T=\frac{Y-1+R}{K}.
\]
Then $T\sim\Unif[0,1]$.  Also, $T$ almost surely determines $Y$, and the
extra within-bin randomness $R$ is independent of $Z$ once $Y$ is known.  Hence
\[
  \II(Z;T)=\II(Z;Y).
\]
It remains to compare threshold entropies.  Fix a threshold $t\in[0,1]$, and
let $\ell$ be the bin containing $t$.  Put
\[
  A=\1_{\{Y\le \ell-1\}},
  \qquad
  W=\1_{\{T\le t\}}.
\]
The bits $A$ and $W$ can differ only when $Y=\ell$, an event of probability
$1/K$.  Therefore
\[
  \Pbb(A\ne W)\le 1/K.
\]
Since $W$ is determined by $A$ together with the error bit $A\oplus W$,
\[
\begin{aligned}
  \HH(W\mid Z)
  &\le \HH(A\mid Z)+\HH(A\oplus W) \\
  &\le \HH(A\mid Z)+\hbin(1/K).
\end{aligned}
\]
Here we used that $\hbin$ is increasing on $[0,1/2]$ and $K\ge2$.  Taking the
supremum over $t$ gives
\[
  h(Z,T)
  \le
  \max_{1\le k<K}\HH(\1_{\{Y\le k\}}\mid Z)
  +\hbin(1/K).
\]
Thus the continuous witness obtained from $(Z,Y)$ has objective at most
\[
  \II(Z;Y)
  +2\max_{1\le k<K}\HH(\1_{\{Y\le k\}}\mid Z)
  +2\hbin(1/K)
  =\mathsf B_K(Z,Y)+2\hbin(1/K).
\]
Taking the infimum over finite witnesses $(Z,Y)$ proves
$\gamma\le \gamma_K+2\hbin(1/K)$.
\end{proof}

We now describe the finite certificate used to lower bound $\gamma_K$.  Let
\begin{equation*}
  \triangle_K=\bigl\{r=(r_1,\ldots,r_K)\in[0,1]^K:
      \sum_{j=1}^K r_j=1\bigr\}.
\end{equation*}
For $r\in\triangle_K$, let
\begin{equation*}
  C_k(r)=\sum_{j=1}^k r_j,
  \qquad 1\le k<K.
\end{equation*}
We use the convention $0\log_2 0=0$, so that
$r_j\log_2(Kr_j)$ is interpreted as $0$ when $r_j=0$.

\begin{proposition}
\label{prop:finite-dual-simplex-certificate}
Suppose there exist numbers $\lambda_1,\ldots,\lambda_{K-1}\ge0$, a vector
$a=(a_1,\ldots,a_K)\in\R^K$, and a number $E$ such that
\begin{equation}\label{eq:finite-dual-load}
  \sum_{k=1}^{K-1}\lambda_k\le2
\end{equation}
and, for every $r\in\triangle_K$,
\begin{equation}\label{eq:finite-dual-simplex-cert}
  \sum_{j=1}^K r_j\log_2(Kr_j)
  +\sum_{k=1}^{K-1}\lambda_k\hbin(C_k(r))
  \ge
  E+\sum_{j=1}^K a_j(r_j-1/K).
\end{equation}
Then $\gamma_K\ge E$.  Consequently $\gamma\ge E$.
\end{proposition}

\begin{proof}
Let $(Z,Y)$ be an arbitrary coupling with $Y\sim\Unif([K])$.  Given $Z$, let
\begin{equation*}
  r_j(Z)=\Pbb(Y=j\mid Z),
  \qquad j=1,\ldots,K.
\end{equation*}
This is a random point of $\triangle_K$.  Since $Y$ is uniform on $[K]$,
\begin{equation*}
  \Ebb\left[r_j(Z)\right]=1/K
  \qquad\text{for every }j.
\end{equation*}
Average \eqref{eq:finite-dual-simplex-cert} over $Z$.  The affine term on the r.h.s.\ vanishes, and we obtain
\begin{equation*}
  \Ebb\left[\sum_{j=1}^K r_j(Z)\log_2(Kr_j(Z))\right]
  +\sum_{k=1}^{K-1}\lambda_k
     \Ebb\left[\hbin(C_k(r(Z)))\right]
  \ge E.
\end{equation*}
The first term is the mutual information, written explicitly as
\begin{equation*}
  \II(Z;Y)
  =\HH(Y)-\HH(Y\mid Z)
  =\Ebb\left[\sum_{j=1}^K r_j(Z)\log_2(Kr_j(Z))\right].
\end{equation*}
Also,
\begin{equation*}
  C_k(r(Z))=\Pbb(Y\le k\mid Z),
\end{equation*}
so
\begin{equation*}
  \Ebb\left[\hbin(C_k(r(Z)))\right]
  =\HH(\1_{\{Y\le k\}}\mid Z).
\end{equation*}
Thus
\begin{equation*}
  \II(Z;Y)
  +\sum_{k=1}^{K-1}\lambda_k
    \HH(\1_{\{Y\le k\}}\mid Z)
  \ge E.
\end{equation*}
By \eqref{eq:finite-dual-load},
\begin{equation*}
  \sum_{k=1}^{K-1}\lambda_k
    \HH(\1_{\{Y\le k\}}\mid Z)
  \le
  2\max_{1\le k<K}\HH(\1_{\{Y\le k\}}\mid Z).
\end{equation*}
Therefore $\mathsf B_K(Z,Y)\ge E$.  Since $(Z,Y)$ was arbitrary, $\gamma_K\ge E$.
The last statement follows from \cref{lem:finite-dual-transfer}.
\end{proof}

One can show via a duality argument that for every $\epsilon$, there is a $(\lambda, a)$ certificate with value at least $\gamma_K-\epsilon$. We will however not use this here, and thus leave out the proof.

Given a candidate certificate, it remains to verify \eqref{eq:finite-dual-simplex-cert}.  Although this is
an inequality on the $(K-1)$-dimensional simplex, its structure allows its verification to be reduced to a one-dimensional dynamic program.  Set
\begin{equation*}
  s_0=0,
  \qquad
  s_j=C_j(r)\quad(1\le j<K),
  \qquad
  s_K=1.
\end{equation*}
Then $r_j=s_j-s_{j-1}$ and
\begin{equation*}
  0=s_0\le s_1\le\cdots\le s_K=1.
\end{equation*}
Define
\begin{equation*}
  \ell_j(x)=x\log_2(Kx)-a_jx,
  \qquad 0\le x\le1.
\end{equation*}
After substituting $r_j=s_j-s_{j-1}$, the certificate inequality is equivalent
to
\begin{equation*}
  \sum_{j=1}^K \ell_j(s_j-s_{j-1})
  +\sum_{k=1}^{K-1}\lambda_k\hbin(s_k)
  +\frac1K\sum_{j=1}^K a_j
  -E
  \ge0
\end{equation*}
for every nondecreasing sequence $(s_j)_{j=0}^K$ with $s_0=0$ and $s_K=1$.

Now define
\begin{equation*}
  V_0(0)=0,
  \qquad
  V_0(x)=+\infty\quad(x\ne0),
\end{equation*}
and, for $1\le j<K$,
\begin{equation*}
  V_j(x)=\lambda_j\hbin(x)+\inf_{0\le y\le x}
  \{V_{j-1}(y)+\ell_j(x-y)\}.
\end{equation*}
The final step is
\begin{equation*}
  V_K(1)=\inf_{0\le y\le1}\{V_{K-1}(y)+\ell_K(1-y)\}.
\end{equation*}
Unwinding the recursion gives
\begin{align*}
&\inf_{r\in\triangle_K}
\biggl[
  \sum_{j=1}^K r_j\log_2(Kr_j)
  +\sum_{k=1}^{K-1}\lambda_k\hbin(C_k(r))
  -\sum_{j=1}^K a_j(r_j-1/K)
\biggr] \\
&\hspace{4cm}
=V_K(1)+\frac1K\sum_{j=1}^K a_j.
\end{align*}
Thus, if $V_K(1)+\frac1K\sum_{j=1}^K a_j$ is at least
$E$, then \eqref{eq:finite-dual-simplex-cert} holds and
\cref{prop:finite-dual-simplex-certificate} proves $\gamma\ge E$.

The numerical computation thus has two stages.  First, we search for good
coefficients $(\lambda,a)$ to maximize  $V_K(1)+\frac1K\sum_{j=1}^K a_j$. In theory, a certificate can be computed yielding the exact infimum up to some arbitrarily chosen precision. However, the computational cost is prohibitive, so we resort to looking for $(\lambda,a)$ heuristically. 

In the second stage (the actual proof), the dynamic program is verified, pessimistically rounding intermediate values to ensure the final bound is valid. To run the dynamic program, the continuous variable $x\in [0,1]$ is discretized into $M$ discrete bins (for some large value of $M$) and the worst-case values inside the bins are bounded using the values at the endpoints and the concavity/convexity of the functions involved. Larger values of $M$ result in more computation, but smaller loss in this process and thus a better final bound. 

The certificate we use for $K=400$ and $M=300\ 000$ gives the following.

\begin{theorem}\label{thm:gamma_lb}
    \begin{equation*}
   \inf_{0<\sigma\le1}\{2\sigma-\Dent_\sigma\} = \gamma> E_{400} > 1.669845.
    \end{equation*}
    Equivalently,
    \begin{equation*}
      \inf_{1<S\le2}\frac{S^2}{D_S}> 2^{1.669845} > 3.1818.
    \end{equation*}
\end{theorem}

The finite-continuous comparison in \cref{lem:finite-dual-transfer} also shows
that the finite computation converges to the continuous optimum: indeed
$\gamma_K\le \gamma\le \gamma_K+2\hbin(1/K)$, and the error term tends to zero. In theory, the described approach can thus be used to compute a bound which is tight up to an arbitrarily chosen precision.
However, the coefficients $(\lambda,a)$ are expensive to optimize exactly, and this explicit error term is a worst-case bound and decreases rather
slowly. For the upper side of the final numerical interval it is therefore
better to use the explicit witnesses from the previous section.  Since
\cref{thm:two-state-corridor-close} gives $\Dent_\sigma\ge a_\sigma-1$ for
$1/2\le\sigma\le1$, we have
\begin{equation*}
  \gamma
  =\inf_{0<\sigma\le1}\{2\sigma-\Dent_\sigma\}
  \le
  \inf_{1/2\le\sigma\le1}\{2\sigma-a_\sigma+1\}.
\end{equation*}

The right-hand side is at most
\begin{equation*}
  1.669858,
  \qquad
  2^{1.669858}<3.1818328,
\end{equation*}
attained near $\sigma=0.515617$.  We conjecture that the two quantities
coincide exactly, that is,
\begin{equation*}
  \inf_{0<\sigma\le1}\{2\sigma-\Dent_\sigma\}
  =\inf_{1/2\le\sigma\le1}\{2\sigma-a_\sigma+1\}.
\end{equation*}

\subparagraph*{Computing upper bounds for $\Dent_\sigma$}

The lower bound of $2\sigma-\Dent_\sigma > 1.669845$ from the previous theorem immediately gives the following upper bound on $\Dent_\sigma$.
\[\Dent_\sigma \leq 2\sigma-1.669845.\]
This is essentially tight near $\sigma \approx 0.515617$. However, this bound gets worse the further we are from this value of $\sigma$. 

The lower bound on $1/D_S$ given in \cite{DK} (where this quantity is denoted as $P_S$), when translated to our setting using Theorem~\ref{thm:comb_equals_entropy}, gives the following.
\begin{proposition}[\cite{DK}]
For every $0<\sigma\leq 1$, and every integer $k \geq 0$,
\[
    \Dent_\sigma \leq k\sigma-\log_2(k+1).
\]
\end{proposition}
This gives better bounds when $\sigma$ is near $0$ or $1$, but is overall very loose.

One way to get better upper bounds of this nature is to essentially look for lower bounds on $\Lambda\sigma-\Dent_\sigma$, for different values of $\Lambda > 0$. This can be done by trivial modifications of the method outlined above (in short, replace the constraint \eqref{eq:finite-dual-load} with $\sum_{k=1}^{K-1}\lambda_k\leq \Lambda$). A lower bound of $\mu_\Lambda$ on this value then gives
\[\Dent_\sigma \leq \Lambda\sigma-\mu_\Lambda.\]
Choosing $\Lambda$ smaller makes the resulting bound better for larger values of $\sigma$, while choosing a larger value for $\Lambda$ favors good bounds for smaller $\sigma$.

Below are some recorded bounds for some different values of $\Lambda$, computed with $K=100$ and $M=50\ 000$ (except for the case $\Lambda=2$ which was computed earlier).\footnote{The certificates, as well as the code used to verify them, will be made available on the author's GitHub repository.}

\begin{proposition}\label{prop:dent_ub}
    We have $\Dent_\sigma \leq \Lambda\sigma-\mu_\Lambda$, for each of the following $(\Lambda,\mu_\Lambda)$ pairs:
    \begin{align*}
        & (\Lambda=1.1, \mu_\Lambda=1.0925), &
        & (\Lambda=1.2, \mu_\Lambda=1.1778), &
        & (\Lambda=1.3, \mu_\Lambda=1.2554), &
        & (\Lambda=1.6, \mu_\Lambda=1.4507), \\
        & (\Lambda=2,\mu_\Lambda=1.6698), &
        & (\Lambda=2.5,\mu_\Lambda=1.9189), &
        & (\Lambda=3,\mu_\Lambda=2.1276), &
        & (\Lambda=4,\mu_\Lambda=2.4737), \\
        & (\Lambda=5,\mu_\Lambda=2.7515), &
        & (\Lambda=7,\mu_\Lambda=3.1855), &
        & (\Lambda=10,\mu_\Lambda=3.6605), &
        & (\Lambda=13,\mu_\Lambda=4.0160), \\
        & (\Lambda=16,\mu_\Lambda=4.3001), &
        & (\Lambda=20,\mu_\Lambda=4.6072), &
        & (\Lambda=30,\mu_\Lambda=5.1659), &
        & (\Lambda=60,\mu_\Lambda=6.0930).
    \end{align*}
\end{proposition}

\section{An application to fibres in the Boolean lattice: Proof of Theorem~\ref{thm:fibre}}\label{sec8}
Here we give another consequence of our bounds on $\Dent_\sigma$, beyond the algorithmic applications to TSP and other permutation problems. This proves Theorem~\ref{thm:fibre}.

Let $B_n = 2^{[n]}$ denote the $n$-dimensional Boolean lattice, with the sets composing them being compared by inclusion. Recall that an antichain of $B_n$ is a subfamily of $B_n$ consisting of sets which are pairwise incomparable. An antichain is maximal if it cannot be extended to a larger antichain by adding a set. A fibre of $B_n$ is a family $\mathcal A\subseteq B_n$ that intersects every maximal antichain. 

We let $\mathrm{MF}_n$ denote the size of a smallest fibre in $B_n$. It is easy to show that $\mathrm{MF}_n \leq O(2^{n/2})$: consider the family of all sets comparable to any given set of the middle layer of $B_n$; this is a fibre and has size $O(2^{n/2})$. 

Lonc and Rival~\cite{LoncRival1987} conjectured that this is tight. Duffus, Sands and Winkler~\cite{DuffusSW90} made some progress towards that conjecture, by showing $\mathrm{MF}_n \geq \Omega(1.25^{n})$. This was later improved by {\L}uczak to $\mathrm{MF}_n \geq \Omega(2^{n/3})$, as reported by Duffus and Sands~\cite{DuffusSands2001}. There has been no further progress since.

We exploit our bounds on $\Dent_\sigma$ to improve on this lower bound, showing the following.

\restateFIBRE*

For the proof, we need the following additional fact from \cite{DuffusSW90}: every fibre of $B_n$ contains at least $n!/2^{n-1}$ maximal chains.

\begin{proof}
    Fix $\sigma$, and assume $D^\ell_\sigma \leq -1$ (equivalently, $D_{2^\sigma} \leq \frac{1}{2}$).
    Let $\mathcal A$ be a fibre of $B_n$. Since $\mathcal A$ contains at least $n!/2^{n-1}$ maximal chains, we have
    \[
    D(\mathcal A) \geq \frac{1}{2}2^{1/n} > \frac{1}{2}.
    \]
     Then if $S(\mathcal A) \leq 2^{\sigma}$, by definition of $D_{2^\sigma}$ we would have $D(\mathcal A) \leq D_{2^\sigma} \leq \frac{1}{2}$, a contradiction.

    Thus $S(\mathcal A) \geq 2^{\sigma}$, i.e.\ $|\mathcal A|\geq 2^{\sigma n}$.

    For the second statement, it suffices to note that taking $\Lambda=3$ in \cref{prop:dent_ub} gives
    \[
        D^\ell_\sigma = \Dent_\sigma \leq 3\sigma-2.1276.
    \]
    Thus, $D^\ell_\sigma \leq -1$ for any $\sigma \leq \frac{2.1276-1}{3} = 0.37586\ldots$, and the statement follows.
\end{proof}
Note that this approach cannot be used to prove bounds above $2^{0.38531n}$, as we know from Theorem~\ref{thm:closed-bound-all-S} that $\Dent_{0.38531} > -1$. Different ideas would thus be needed to prove the conjecture of Lonc and Rival.

\section{Conclusion}\label{sec:concl}

We gave an essentially tight tradeoff between the normalized size and the normalized chain-density of set systems. This provides a fairly complete answer to the extremal combinatorics question raised by Johnson, Leader, and Russell, and implies an improved space-time tradeoff for TSP (and related permutation problems) in the recently introduced framework of~\cite{ANW, DK}.

The special case of set systems arising as the family of downsets of a partial order (here, the maximal chains correspond to linear extensions of the same partial order) remains open; while the barrier side of our bounds still hold in this setting, our set system constructions are not of this restricted form. In this setting, the construction of Koivisto and Parviainen~\cite{KoivistoParviainen2010} was improved by Ameli, Nederlof, and Wang~\cite{ANW}. Their construction remains the best to date, leaving the gap $3.1818 \leq \gamma \leq 3.7493$ in the partial order case. For \emph{regular bipartite posets}, \cite{ANW} show a lower bound of 3.6, implying a clear separation of this special case from the general set systems attaining an optimal tradeoff.

Solving the TSP in time $c^n$ for $c<2$, or providing some argument for why this would not be possible, remains, as one of the landmark questions of computer science, entirely open. Likewise, the question of whether a polynomial-space algorithm exists with running time $c^n$ for $c<4$ remains unsolved. More modestly, obtaining an algorithm that improves our space-time tradeoff $\cT \cdot \cS \approx 3.1819$  via different techniques remains an interesting research direction.

\appendix

\newpage

\small
\bibliographystyle{alphaurl}
\bibliography{main}
\end{document}

%% file: section_3_2_arx.tex
\subsection{From $(Z,T)$ witnesses to set system witnesses}
We now prove the reverse direction of \cref{thm:comb_equals_entropy}.  The
point is that a witness $(Z,T)$ can be converted into large families of
permutations whose prefix sets form the desired set systems.  We first carry
this out when $Z$ has finite range.  The general case is then obtained by
coarsening $Z$ into finitely many values, losing only an arbitrarily small
amount in the threshold entropy and losing nothing in mutual information.

\begin{proposition}\label{prop:block-profile-approx}
Let $(Z,T)$ be an admissible pair such that $Z$ takes finitely many values and
$d=\II(Z;T)<\infty$.  Let $h=h(Z,T)$.  For every $\eta>0$ there are set systems
$\F_n\subseteq2^{[n]}$ with $C(\F_n)>0$ such that
\begin{align*}
  \mathrm{(i)}~~ \limsup_{n\to\infty}S^{\ell}(\F_n) & \le h+\eta \\
   \mathrm{(ii)}~~\liminf_{n\to\infty} D^{\ell}(\F_n) & \ge  -d.
\end{align*}
\end{proposition}

\begin{proof}
Write the values of $Z$ as $1,\ldots,s$, and set
\[ 
  p_i=\Pbb(Z=i)>0,
  \qquad
  F_i(t)=\Pbb(T\le t\mid Z=i).
\]
Since $T$ is uniform on $[0,1]$, the conditional distribution functions obey
\begin{equation}\label{eq:mixture-uniform-main}
  \sum_{i=1}^s p_iF_i(t)= \Pbb(T\leq t) = t,
  \qquad 0\le t\le1.
\end{equation}
In particular, each $F_i$ is continuous: if $F_i$ had a discontinuity at some point,
then so would $t\to \Pbb(T\leq t)$ (because all $F_i$ are non-negative), contradicting the uniformity of
$T$. Because each $F_i$ is also trivially bounded and monotone, it is uniformly continuous.

For later use, note that
\begin{equation}\label{eq:d-h-finite-main}
  h = h(Z,T)
  =\sup_{0\le t\le1}\sum_{i=1}^s p_i\hbin(F_i(t)),
\end{equation}
because, conditional on $Z=i$, the bit $\1_{\{T\le t\}}$ is Bernoulli with
parameter $F_i(t)$.

Choose a partition
\[
  0=t_0<t_1<\cdots<t_L=1
\]
fine enough so that the following holds.  Whenever $t_{a-1}\le t\le t_a$, the values
$F_i(t)$ and $F_i(t_{a-1})$ are close enough, for every $i$, that any vector
$x=(x_1,\ldots,x_s)$ satisfying
\[
  F_i(t_{a-1})\le x_i\le F_i(t_a)
  \qquad\text{for all }i
\]
obeys
\begin{equation}\label{eq:rectangle-entropy-control}
  \sum_{i=1}^s p_i\hbin(x_i)
  \le h+\frac{\eta}{3}.
\end{equation}
Such a partition exists from the uniform continuity of the finitely many functions
$F_i$ and of the binary entropy function: after refining the partition, each
$x_i$ can be made to be arbitrarily close to $F_i(t_{a-1})$, and the value at the left endpoint is at
most $h$ by \eqref{eq:d-h-finite-main}.

Let $A$ be the interval of the partition containing $T$, so $A=a$ when
$t_{a-1}<T\le t_a$.  Define
\[
  \Delta_a=t_a-t_{a-1},
  \qquad
  q_{ia}=p_i\bigl(F_i(t_a)-F_i(t_{a-1})\bigr).
\]
Thus $q_{ia}$ is the joint probability of the event $\{Z=i,A=a\}$.  Its row
and column sums are
\[
  \sum_a q_{ia}=p_i,
  \qquad
  \sum_i q_{ia}=\Delta_a,
\]
where the second identity is \eqref{eq:mixture-uniform-main} applied at
$t_a$ and $t_{a-1}$.

We now turn this finite table into a set system.  For each $n$, choose integers
$n_{ia}\ge0$ with
\[
  \frac{n_{ia}}{n}\to q_{ia}
  \qquad\text{for every }i,a,
\]
and with total sum $\sum_{i,a}n_{ia}=n$.  Let
\[
  n_i=\sum_a n_{ia},
  \qquad
  m_a=\sum_i n_{ia}.
\]
Then $n_i/n\to p_i$ and $m_a/n\to\Delta_a$.  Partition the ground set into
classes
\[
  [n]=B_1\sqcup\cdots\sqcup B_s,
  \qquad |B_i|=n_i,
\]
and partition the positions of a permutation into consecutive blocks
\[
  R_a=\{m_1+\cdots+m_{a-1}+1, m_1+\cdots+m_{a-1}+2, \ldots,m_1+\cdots+m_a\},
  \qquad |R_a|=m_a.
\]
Let $\Pi_n$ be the set of permutations $\pi$ of $[n]$ such that, for every
$i$ and $a$,
\begin{equation*}
  \#\{j\in R_a:\pi(j)\in B_i\}=n_{ia}.
\end{equation*}
For $\pi\in\Pi_n$ and $0\le k\le n$, let
\[
  P_k(\pi)=\{\pi(1),\ldots,\pi(k)\},
  \qquad P_0(\pi)=\varnothing.
\]
Define
\begin{equation*}
  \F_n=\{P_k(\pi):\pi\in\Pi_n,\ 0\le k\le n\}.
\end{equation*}
Every permutation in $\Pi_n$ gives a maximal chain contained in $\F_n$.  Hence
\begin{equation*}
  C(\F_n)\ge |\Pi_n|.
\end{equation*}

We first estimate the chain density.  A permutation in $\Pi_n$ can be specified as follows. For each $a$ and each $i$, choose which $n_{ia}$ positions $j\in R_a$ have $\pi(j)\in B_i$. Then, for each $i$, assign the elements of $B_i$ to
all positions of that type.  This gives
\begin{equation*}
  |\Pi_n|=\frac{\prod_i n_i!\prod_a m_a!}{\prod_{i,a}n_{ia}!}.
\end{equation*}
By Stirling's formula, with the convention $0\log 0=0$,
\begin{align*}
  \lim_{n\to\infty}\frac1n\log_2\frac{n!}{|\Pi_n|}
  &=\sum_{i,a}q_{ia}\log_2\frac{q_{ia}}{p_i\Delta_a} \\
  &=\II(Z;A).
\end{align*}

Since $A$ is a function of $T$, data processing gives
$\II(Z;A)\le\II(Z;T)=d$.  Together with $C(\F_n)\ge |\Pi_n|$, this implies
\[
  \liminf_{n\to\infty}D^{\ell}(\F_n)
  =\liminf_{n\to\infty}\frac1n\log_2\frac{C(\F_n)}{n!}
  \ge -d,
\]
which proves $\mathrm{(ii)}$.

It remains to estimate the size of $\F_n$.  Let $X\in\F_n$ be a prefix whose
size lies in the position block $R_a$.  For each type $i$, let
\[
  x_i(X)=\frac{|X\cap B_i|}{n_i}.
\]
All elements placed in earlier position blocks have already appeared in the
prefix, and no element placed in later position blocks has appeared yet.  Hence
\begin{equation*}
  \frac{\sum_{b<a}n_{ib}}{n_i}
  \le x_i(X)
  \le
  \frac{\sum_{b\le a}n_{ib}}{n_i}.
\end{equation*}
The two endpoints converge to $F_i(t_{a-1})$ and $F_i(t_a)$, respectively.  By
\eqref{eq:rectangle-entropy-control} and the convergence $n_i/n\to p_i$, it
follows that, for all sufficiently large $n$, the following holds
\begin{equation}\label{eq:profile-entropy-control-n}
  \sum_{i=1}^s\frac{n_i}{n}\hbin(x_i(X))
  \le h+\frac{2\eta}{3}.
\end{equation}

For a fixed profile $(k_1,\ldots,k_s)$, where $k_i=|X\cap B_i|$, the number of
sets with that profile is at most
\[
  \prod_i\binom{n_i}{k_i}
  \le
  2^{\sum_i n_i\hbin(k_i/n_i)},
\]
using the standard binomial estimate
\begin{equation*}
  \binom{N}{K}\le 2^{N\hbin(K/N)}.
\end{equation*}
Thus \eqref{eq:profile-entropy-control-n} shows that each profile $(k_1,\ldots,k_s)$ contributes
at most $2^{n(h+2\eta/3)}$ sets.  There are at most $\prod_i(n_i+1)$ possible
profiles, which is polynomial in $n$.  Therefore
\[
  \frac1n\log_2|\F_n|
  \le h+\frac{2\eta}{3}+o(1).
\]
Taking the limsup and using the remaining slack in $\eta$ proves
$\mathrm{(i)}$.
\end{proof}

We next remove the assumption that $Z$ has finite range.

\begin{lemma}\label{lem:finite-valued-reduction}
Let $(Z,T)$ be admissible with $\II(Z;T)<\infty$.  For every $\rho>0$ there is
a finite-valued random variable $Z_\rho$, which is a function of $Z$, such that

\begin{align*}
  \mathrm{(i)} ~~ \II(Z_\rho;T) & \le\II(Z;T), \\
   \mathrm{(ii)}  ~~h(Z_\rho,T) & \le h(Z,T)+\rho.
\end{align*}
\end{lemma}

\begin{proof}
Fix $\varepsilon>0$ and choose a grid
\[
  0=t_0<t_1<\cdots<t_L=1
\]
such that the distance between consecutive points is at most $\varepsilon$.  Let $B_t=\1_{\{T\le t\}}$.  If
$t_{a-1}\le t\le t_a$, then $B_t$ and $B_{t_a}$ differ only when
$t<T\le t_a$, an event of probability at most $\varepsilon$.  Hence, for any
random variable $W$,
\begin{equation}\label{eq:grid-threshold-reduction}
  \HH(B_t\mid W)
  \le \HH(B_{t_a}\mid W)+\hbin(\varepsilon),
\end{equation}
because $B_t$ is determined by $B_{t_a}$ together with the error bit
$B_t\oplus B_{t_a}$.

For the grid points $t_0,\ldots,t_{L}$, fix versions of the conditional
probabilities
\[
  g_a(Z)=\Pbb(T\le t_a\mid Z).
\]
Partition the cube $[0,1]^{L+1}$ into finitely many boxes of side length at
most $\delta$, and let $Z_\rho$ be the index of the box containing
$(g_0(Z),\ldots,g_{L}(Z))$.  Then $Z_\rho$ is finite-valued and is a function
of $Z$, so $(i)$ follows from data processing.

Let
\[
  \omega_{\hbin}(\delta)
  =
  \sup\{ |\hbin(p)-\hbin(q)| : |p-q|\le \delta,\ p,q\in[0,1] \}
\]
be a modulus of continuity for the binary entropy function.  Fix a grid point
$t_a$, and write
\[
  g_a(Z)=\Pbb(T\le t_a\mid Z).
\]
Since $Z_\rho$ is obtained from $Z$, conditioning first on $Z$ and then
averaging over the cell of $Z_\rho$ gives
\[
  \Pbb(T\le t_a\mid Z_\rho)
  =
  \Ebb(g_a(Z)\mid Z_\rho).
\]
Now consider one cell of the partition.  By construction, the $a$-th coordinate
$g_a$ varies by at most $\delta$ inside that cell.  Therefore the average
$\Ebb(g_a(Z)\mid Z_\rho)$, which is an average of values of $g_a$ from the same
cell, is also within $\delta$ of each value $g_a(Z)$ in that cell.  Thus,
almost surely,
\[
  \left|
    \Pbb(T\le t_a\mid Z_\rho)-g_a(Z)
  \right|
  \le \delta.
\]
Applying the modulus of continuity of $\hbin$, we get the point-wise bound
\[
  \hbin\!\left(\Pbb(T\le t_a\mid Z_\rho)\right)
  \le
  \hbin(g_a(Z))+\omega_{\hbin}(\delta).
\]
Taking expectations gives
\[
  \HH(B_{t_a}\mid Z_\rho)
  \le
  \HH(B_{t_a}\mid Z)+\omega_{\hbin}(\delta).
\]
Finally,
\[
  \HH(B_{t_a}\mid Z)\le h(Z,T),
\]
because $h(Z,T)$ is the supremum of $\HH(B_t\mid Z)$ over all thresholds $t$.
Hence, for every grid point $t_a$,
\[
  \HH(B_{t_a}\mid Z_\rho)
  \le
  h(Z,T)+\omega_{\hbin}(\delta).
\]

 Now apply \eqref{eq:grid-threshold-reduction} with
$W=Z_\rho$ to get
\begin{align*}
    \HH(B_t\mid Z_\rho)
    &\leq \HH(B_{t_a}\mid Z_\rho)+\hbin(\varepsilon)\\
    &\leq h(Z,T)+\omega_{\hbin}(\delta)+\hbin(\varepsilon).
\end{align*}

Taking the supremum over all $t$:
\[
  h(Z_\rho,T)
  \le h(Z,T)+\omega_{\hbin}(\delta)+\hbin(\varepsilon).
\]
Choosing $\delta$ and $\varepsilon$ so that the last two terms sum to at most
$\rho$ proves $\mathrm{(ii)}$.
\end{proof}

\begin{proof}[Proof of \cref{thm:comb_equals_entropy}]
We first prove $D^{\ell}_{\sigma}\le\Dent_\sigma$.  By
\cref{lem:set-to-entropy-witness}, every family $\F$ with
$S^{\ell}(\F)\le\sigma$ gives an admissible pair $(Z,T)$ with
$h(Z,T)\le\sigma$ and
\[
  -\II(Z;T)=D^{\ell}(\F).
\]
Taking the supremum over all such families gives
$D^{\ell}_{\sigma}\le\Dent_\sigma$.

For the reverse inequality, fix $0<\tau<\sigma$ and an admissible pair
$(Z,T)$ with $h(Z,T)\le\tau$ and $\II(Z;T)<\infty$.  Apply
\cref{lem:finite-valued-reduction} with a small parameter $\rho>0$, and then
apply \cref{prop:block-profile-approx} to the finite-valued pair
$(Z_\rho,T)$.  If $\rho$ and the parameter $\eta$ in the proposition are chosen
so that $\tau+\rho+\eta\le\sigma$, we obtain set systems $\F_n$ with
\[
  \limsup_n S^{\ell}(\F_n)\le\sigma
  \qquad\text{and}\qquad
  \liminf_n D^{\ell}(\F_n)
  \ge -\II(Z_\rho;T)
  \ge -\II(Z;T).
\]
Thus $D^{\ell}_{\sigma}\ge -\II(Z;T)$.  Taking the supremum over all witnesses
with $h(Z,T)\le\tau$ gives
\[
  D^{\ell}_{\sigma}\ge \Dent_\tau.
\]
Finally let $\tau$ approach $\sigma$ from below and use the continuity of $\Dent$ from
\cref{lem:Dent-concave}. This proves
$D^{\ell}_{\sigma}\ge\Dent_\sigma$, and hence
$D^{\ell}_{\sigma}=\Dent_\sigma$.

The formula involving $S^2/D_S$ is then just a change of variables.  Since
$S=2^\sigma$ and $D^{\ell}_{\sigma}=\log_2 D_{2^\sigma}$,
\[
  \log_2\frac{S^2}{D_S}=2\sigma-D^{\ell}_{\sigma}=2\sigma-\Dent_\sigma.
\]
The remaining equalities in the statement follow directly from the definition
of $\Dent_\sigma$: for a fixed admissible pair $(Z,T)$, the smallest feasible
choice of $\sigma$ is $h(Z,T)$, and the expression $2\sigma+\II(Z;T)$ only
increases when $\sigma$ is made larger.
\end{proof}